\pdfoutput=1
\PassOptionsToPackage{unicode}{hyperref}
\PassOptionsToPackage{hyphens}{url}
\PassOptionsToPackage{dvipsnames,svgnames*,x11names*}{xcolor}
\documentclass[
]{article}
\usepackage{lmodern}
\usepackage{amsmath}
\usepackage{ifxetex,ifluatex}
\ifnum 0\ifxetex 1\fi\ifluatex 1\fi=0 
  \usepackage[T1]{fontenc}
  \usepackage[utf8]{inputenc}
  \usepackage{textcomp} 
  \usepackage{amssymb}
\else 
  \usepackage{unicode-math}
  \defaultfontfeatures{Scale=MatchLowercase}
  \defaultfontfeatures[\rmfamily]{Ligatures=TeX,Scale=1}
\fi
\IfFileExists{upquote.sty}{\usepackage{upquote}}{}
\IfFileExists{microtype.sty}{
  \usepackage[]{microtype}
  \UseMicrotypeSet[protrusion]{basicmath} 
}{}
\makeatletter
\@ifundefined{KOMAClassName}{
  \IfFileExists{parskip.sty}{%
    \usepackage{parskip}
  }{
    \setlength{\parindent}{0pt}
    \setlength{\parskip}{6pt plus 2pt minus 1pt}}
}{
  \KOMAoptions{parskip=half}}
\makeatother
\usepackage{xcolor}
\IfFileExists{xurl.sty}{\usepackage{xurl}}{} 
\IfFileExists{bookmark.sty}{\usepackage{bookmark}}{\usepackage{hyperref}}
\hypersetup{
  pdftitle={A closed form scale bound for the (\textbackslash epsilon, \textbackslash delta)-differentially private Gaussian Mechanism valid for all privacy regimes},
  pdfauthor={Staal A. Vinterbo},
  colorlinks=true,
  linkcolor=Red,
  filecolor=Maroon,
  citecolor=Blue,
  urlcolor=orange,
  pdfcreator={LaTeX via pandoc}}
\providecommand{\tightlist}{%
  \setlength{\itemsep}{0pt}\setlength{\parskip}{0pt}}
\usepackage[a4paper]{geometry}
\usepackage[small,bf]{caption}
\usepackage{graphicx}
\usepackage{subcaption}
\usepackage{amsmath}
\usepackage{amsthm}
\begingroup
    \makeatletter
    \@for\theoremstyle:=definition,remark,plain\do{%
        \expandafter\g@addto@macro\csname th@\theoremstyle\endcsname{%
            \addtolength\thm@preskip\parskip
            }%
        }
\endgroup
\usepackage{color}

\theoremstyle{plain}
\newtheorem{theorem}{Theorem}
\newtheorem{lemma}[theorem]{Lemma}

\newtheorem{corollary}{Corollary}
\theoremstyle{definition}
\newtheorem{definition}{Definition}
\newtheorem{remark}{Remark}

\usepackage[notcite,notref,final]{showkeys}

\DeclareMathOperator{\sign}{sign}

\usepackage{stmaryrd}

\makeatletter
\newcommand{\pushright}[1]{\ifmeasuring@#1\else\omit\hfill$\displaystyle#1$\fi\ignorespaces}
\newcommand{\pushleft}[1]{\ifmeasuring@#1\else\omit$\displaystyle#1$\hfill\fi\ignorespaces}
\makeatother

\newcommand{\hideFromPandoc}[1]{#1}
\hideFromPandoc{

}

\usepackage{multicol}

\usepackage{tcolorbox}
\usepackage{needspace} 

\usepackage{proof-at-the-end}
\ifluatex
  \usepackage{selnolig}  
\fi
\usepackage[]{biblatex}
\title{A closed form scale bound for the
\((\epsilon, \delta)\)-differentially private Gaussian Mechanism valid
for all privacy regimes\thanks{This work has been funded in part by
Innlandet Fylkeskommune, as well as Research Council of Norway grants
308904 and 288856. Thanks to Stephan Dreiseitl and Jerome Le Ny for
helpful discussions.}}
\author{Staal A.
Vinterbo\footnote{Department of Information Security and Communication Technology, Norwegian University of Science and Technology.}}
\date{}

\begin{document}
\maketitle
\begin{abstract}
The standard closed form lower bound on \(\sigma\) for providing
\((\epsilon, \delta)\)-differential privacy by adding zero mean Gaussian
noise with variance \(\sigma^2\) is
\(\sigma > \Delta\sqrt {2}(\epsilon^{-1}) \sqrt {\log \left( 5/4\delta^{-1} \right)}\)
for \(\epsilon \in (0,1)\). We present a similar closed form bound
\(\sigma \geq \Delta (\epsilon\sqrt{2})^{-1} \left(\sqrt{az+\epsilon} + s\sqrt{az}\right)\)
for \(z=-\log(4\delta(1-\delta))\) and \((a,s)=(1,1)\) if
\(\delta \leq 1/2\) and \((a,s)=(\pi/4,-1)\) otherwise. Our bound is
valid for all \(\epsilon > 0\) and is always lower (better). We also
present a sufficient condition for \((\epsilon, \delta)\)-differential
privacy when adding noise distributed according to even and log-concave
densities supported everywhere.
\end{abstract}

\newcommand{\xqed}[1]{%
  \leavevmode\unskip\penalty9999 \hbox{}\nobreak\hfill
  \quad\hbox{\ensuremath{#1}}}

\newcommand{\eprop}{}

\newcommand{\mc}[1]{\ensuremath{\mathcal{#1}}}
\newcommand{\mb}[1]{\ensuremath{\mathbf{#1}}}
\newcommand{\mbb}[1]{\ensuremath{\mathbb{#1}}}

\newcommand{\xcirc}[1]{\vcenter{\hbox{$#1\circ$}}}
\newcommand{\fcomp}{\mathbin{\mathchoice
  {\xcirc\scriptstyle}
  {\xcirc\scriptstyle}
  {\xcirc\scriptscriptstyle}
  {\xcirc\scriptscriptstyle}
}}

\newcommand{\R}{\ensuremath{\mathbb{R}}}

\newcommand{\X}{\ensuremath{\mathcal{X}}}
\newcommand{\Y}{\ensuremath{\mathcal{Y}}}
\newcommand{\cov}{\operatorname{cov}}
\newcommand{\var}{\operatorname{Var}}
\newcommand{\E}{\operatorname{E}}
\newcommand{\tr}{\operatorname{Tr}}
\newcommand{\strings}{\ensuremath{\mathbb{S}}}

\newcommand{\mult}{\operatorname{Mult}}
\newcommand{\Bet}{\ensuremath{\mathcal{B}}}
\newcommand{\Lap}{\ensuremath{\mathcal{L}}}
\newcommand{\gam}{\operatorname{Gamma}}
\newcommand{\Exp}{\operatorname{Exp}}
\newcommand{\bit}{\mbb{B}}
\newcommand{\biti}{\ensuremath{\bit^\infty}}
\newcommand{\ms}[1]{\ensuremath{\mathit{#1}}}
\newcommand{\deq}{\ensuremath{\stackrel{\tiny d}{=}}}
\newcommand{\dto}{\ensuremath{\stackrel{\tiny d}{\rightarrow}}}
\newcommand{\rto}{\ensuremath{\stackrel{\tiny\mc{R}}{\rightarrow}}}

\newcommand{\fault}{\ensuremath{\mathrm{fault}}}

\newcommand{\Sketch}{\textsc{Sketch}}
\newcommand{\Assume}{\textsc{Assume}}
\newcommand{\Prove}{\textsc{Prove}}
\newcommand{\Case}{\textsc{Case}}
\newcommand{\Proof}{\textsc{Proof}}
\newcommand{\Let}{\textsc{Let}}

\newcommand{\M}{\mathcal{M}}
\newcommand{\share}{\operatorname{share}}
\newcommand{\csum}{\operatorname{sum}}

\newcommand{\proto}{\textsc{MPLM}}
\newcommand{\corr}{\operatorname{corr}}
\newcommand{\1}{\ensuremath{\mathbf{1}}}
\newcommand{\0}{\ensuremath{\mathbf{0}}}

\newcommand{\diag}{\operatorname{diag}}

\newcommand{\Ri}{\ensuremath{\mbb{R}_{I}}}

\newcommand{\fmod}{\operatorname{fmod}}
\newcommand{\imod}{\operatorname{imod}}

\newcommand{\erf}{\operatorname{erf}}

\newcommand{\mats}{\ensuremath{\mathcal{M}}}

\hypertarget{introduction}{%
\section{Introduction}\label{introduction}}

Differential privacy \autocite{Dwork2006a} is an emerging standard for
individual data privacy. In essence, differential privacy is a bound on
any belief update about an individual on receiving a result of a
differentially private randomized computation. Critical for the utility
of such results is minimizing the random perturbation required for a
given level of privacy.

Formally, let a database \(d\) be a collection of record values from
some set \(V\). Two databases \(d\) and \(d'\) are \emph{neighboring} if
one can be obtained from the other by adding one record. Let
\(\ensuremath{\mathcal{N}}\) be the set of all pairs of neighboring
databases. Then following Dwork et al.
\autocite{Dwork2006a,our-data-ourselves-privacy-via-distributed-noise-generation}
we define differential privacy as follows.

\begin{definition}[$(\epsilon,\delta)$-differential privacy \cite{Dwork2006a,our-data-ourselves-privacy-via-distributed-noise-generation}]

A randomized algorithm \(M\) is called
\((\epsilon,\delta)\)-differentially private if for any measurable set
\(S\) of possible outputs and all
\((d,d') \in \ensuremath{\mathcal{N}}\) \[
\Pr(M(d) \in S) \leq e^{\epsilon}\Pr(M(d') \in S) + \delta,
\] where the probabilities are over randomness used in \(M\). By
\(\epsilon\)-differential privacy we mean \((\epsilon, 0)\)-differential
privacy.

\end{definition}

A standard mechanism for achieving \((\epsilon,\delta)\)-differential
privacy is that of adding zero mean Gaussian noise to a statistic,
called the Gaussian Mechanism. A primary reason for the popularity of
the Gaussian Mechanism is that the Gaussian distribution is closed under
addition. However, Gaussian noise requires \(\delta > 0\), which
represents a relaxation of the stronger \((\epsilon, 0)\)-differential
privacy that is not uncontroversial
\autocite{mcsherryHowManySecrets2017}. On the positive side, a non-zero
\(\delta\) allows, among others, for better composition properties than
\((\epsilon,0)\)-differential privacy \autocite{7883827}. The
exploitation of the composition benefits of using Gaussian noise can be
observed in an application to deep learning by Abadi et
al.\autocite{abadiDeepLearningDifferential2016}.

To achieve \((\epsilon,\delta)\)-differential privacy, the variance
\(\sigma^{2}\) is carefully tuned taking into account the
\emph{sensitivity} \(\Delta\) of the statistic, i.e., the maximum change
in the statistic resulting from adding or removing any individual record
from any database. Of prime importance is to minimize \(\sigma\) while
still achieving \((\epsilon,\delta)\)-differential privacy as higher
\(\sigma\) generally decreases the utility of the now noisy statistic.

In their Theorem A.1
\autocite{dworkAlgorithmicFoundationsDifferential2014}, Dwork and Roth
state that \((\epsilon,\delta)\)-differential privacy is achieved for
\(\epsilon\in (0,1)\) if \begin{align}
\label{eq:drsigma}
\sigma > s(\epsilon, \delta, \Delta)=
{\frac {{\Delta}\,\sqrt {2}}{{\epsilon}}\sqrt {\log
\left( {\frac {5}{4\,{\delta}}} \right) }}.
\end{align} The above bound (\ref{eq:drsigma}) is essentially the
standard closed form used for the Gaussian Mechanism, and we will refer
to it as such in the following. Notably, the restriction
\(\epsilon \in (0,1)\) can present non-obvious pitfalls in addition to
the explicit restriction to privacy regimes with \(\epsilon < 1\). For
example consider the representation of \(\epsilon\) as a function
derived from the standard bound (\ref{eq:drsigma}) (ignoring strict
inequalities) \begin{align}
\label{eq:epsf}
\epsilon(\delta,\sigma,\Delta) &= {\frac {{\Delta}\,\sqrt {2}}{{\sigma}}\sqrt {\log  \left( {\frac 
{5}{4\,{\delta}}} \right) }}.
\end{align} A use of the above function can, for example, be found in
\autocite{pmlr-v89-wang19b}, Section 4. As the magnitude of \(\delta\)
is associated with the failure of guaranteeing strong
\(\epsilon\)-differential privacy, it is sometimes stated that
\(\delta\) should be cryptographically small. Now, the function in
(\ref{eq:epsf}) increases as \(\delta > 0\) decreases, and for fixed
\(\sigma\) and \(\Delta\), even a relatively large \(\delta\) could
result in \(\epsilon(\delta,\sigma,\Delta) \geq 1\), which might not be
obvious. For example, \(\epsilon(10^{-1},1,1) > 2.24\).

\hypertarget{main-contributions}{%
\subsection{Main contributions}\label{main-contributions}}

Our main contributions are twofold:

\begin{enumerate}
\def\labelenumi{\arabic{enumi}.}
\tightlist
\item
  The closed form lower bound on \(\sigma\) for achieving
  \((\epsilon, \delta)\)-differential privacy given in Theorem
  \ref{thm:gaussmech}. Unlike the standard bound, which is defined for
  \(0 < \epsilon < 1\), our bound is valid for all \(\epsilon > 0\). In
  addition it is always better than the standard bound for
  \(0 < \epsilon < 1\).
\item
  The sufficient condition for \((\epsilon, \delta)\)-differential
  privacy for mechanisms that add noise distributed according to an even
  and log-concave density supported everywhere given in Lemma
  \ref{lem:suffcritlc}. The condition is also specialized to the zero
  mean Gaussian distribution in Corollary \ref{cor:gausssuffcrit} and
  the Laplace distribution in Remark \ref{rem:laplace}.
\end{enumerate}

The sufficient condition in point 2.~above is inspired by a sufficient
condition for zero mean Gaussian noise described by Le Ny and Pappas
\autocite{6606817}. We derive their condition from Lemma
\ref{lem:suffcritlc} as Corollary \ref{cor:gausssuffcrit}.

\hypertarget{sec:prelims}{%
\section{A few more preliminaries}\label{sec:prelims}}

We briefly recapitulate known results. In the following, we will let
\(\Phi\) and \(\phi\) denote the standard Gaussian distribution function
and density, respectively.

\begin{definition}

The global sensitivity of a real-valued function \(q\) on databases is
\[
\Delta_q = \max_{(d,d') \in \ensuremath{\mathcal{N}}} |q(d) - q(d')|.
\]

\end{definition}

\begin{lemma}\label{lem:translation}

Let \(X\) be distributed according to density
\(f:\ensuremath{\mathbb{R}}\to \ensuremath{\mathbb{R}}\). For arbitrary
but fixed \(x,y \in \ensuremath{\mathbb{R}}\) we have that
\(\Pr(X + (x-y) \in S) \leq e^{\epsilon} \Pr(X \in S) + \delta\) for all
measurable \(S\) implies
\(\Pr(X + x) \in S) \leq e^{\epsilon} \Pr(X + y \in S) + \delta\) for
all measurable \(S\).

\end{lemma}

\begin{proof}

Follows directly from that if \(S\subseteq \ensuremath{\mathbb{R}}\) is
measurable, so is \(r + S\) for any \(r \in \ensuremath{\mathbb{R}}\),
including \(r = -y\). \qedhere

\end{proof}

\hypertarget{our-closed-form-bound}{%
\section{Our closed form bound}\label{our-closed-form-bound}}

We are now ready to present our main contributions.

\begin{lemma}\label{lem:suffcritlc}

Let \(X\) be a random variable distributed according to even density
\(f(x) = e^{-\psi(x)}\) where
\(\psi : \ensuremath{\mathbb{R}}\to \ensuremath{\mathbb{R}}\) is convex.
Then for a real-valued function \(q\) on databases with global
sensitivity \(\Delta\) and a database \(d\), the mechanism returning a
variate of \(q(d) + s X\) is \((\epsilon, \delta)\)-differentially
private if \begin{align*} 
\Pr \left( X > \frac{x - \Delta}{s}  \right) \leq \delta 
\end{align*} where \[
x \leq \sup \left\{z \mid \psi\left(\frac{z}{s}\right) - \psi\left(\frac{z - \Delta}{s}\right) \leq \epsilon\right\}.
\]

\end{lemma}

\begin{proof}

First, since
\(\psi : \ensuremath{\mathbb{R}}\to \ensuremath{\mathbb{R}}\) is convex,
\(f\) is log-concave and supported everywhere. Let \(f_{sX + x}\) denote
the density of \(sX + x\) for
\(x,s \in \ensuremath{\mathbb{R}}, s > 0\), and recall that
\(f_{sX + x}(w) = s^{-1}f((w - x)/s)\).

Let \(d = x - y\). Due to Lemma \ref{lem:translation} it is sufficient
to show \(\Pr(X + d \in S) \leq e^{\epsilon} \Pr(X \in S) + \delta\) in
order to prove
\(\Pr(X + x \in S) \leq e^{\epsilon} \Pr(X + y \in S) + \delta\).

Since \(f\) is supported everywhere, we now define for
\(x,y\in \ensuremath{\mathbb{R}}\) the likelihood ratio \(r\) as
\begin{align} \label{eq:rdef}
r = \frac{f_{sX + d}}{f_{sX}}.
\end{align}

Let
\(A = \{z \mid r(z) \leq e^{\epsilon}\} = \{z \mid \log(r(z)) \leq {\epsilon}\}\)
and let \(A^{c}\) denote \(A\)'s complement. Then for measurable \(S\)
\begin{align*}
\Pr(s X + d \in S) &= \Pr(s X + d \in S \cap A) + \Pr(s X + d \in S
\cap A^{c}).
\end{align*} Applying (\ref{eq:rdef}), \begin{align*}
\Pr(s X + d \in S \cap A) &\leq \Pr(s X + d \in S ) \\
&= \int_{S} f_{sX}(w)r(w) dw \\
&\leq e^{\epsilon}\int_{S} f_{sX}(w)dw = e^{\epsilon} \Pr(s X \in S).
\end{align*} Furthermore, \begin{align*}
\Pr(s X + d \in S \cap A^{c}) \leq \Pr(s X + d \in A^{c}).
\end{align*} This means that a sufficient condition for
\((\epsilon, \delta)\)-differential privacy is \begin{align}
\label{eq:intcrit}
\Pr(s X + d \in A^{c}) \leq \delta.
\end{align} Since \(f\) is log-concave, \(f\) is unimodal and \(r\) is
monotone (see, e.g., \autocite{saumard2014logconcavity}). Let
\(d \geq 0\), since \(f\) is unimodal \(r\) is non-decreasing and we can
write
\(A = \{w \leq x^{*} \mid x^{*} = \sup \{ z \mid \log(r(z)) \leq\epsilon\}\}\).
Now, let \(d \leq 0\). Then \(r\) is non-increasing and we can write
\(A = \{w \geq x_{*} \mid x_{*} = \inf \{ z \mid \log(r(z)) \leq \epsilon\}\}\).
If \(f\) is also even, we have that \(x_{*} = -x^{*}\), and,
consequently, we need only check (\ref{eq:intcrit}) for either
\(d \leq 0\) or \(d \geq 0\). Let \(d \geq 0\). Since \begin{align*}
\log(r(z)) &= \log \left( \frac{f((z - d)/s)}{f(z/s)} \right) \\
&= \log(f((z - d)/s)) - \log(f(z/s)) \\
&= -\psi((z - d)/s) + \psi(z/s).
\end{align*} We can write
\(x^{*} = \sup \{ z \mid -\psi((z - d)/s) + \psi(z/s) \leq \epsilon\}\),
\(A = \{w \leq x^{*}\}\), and \(A^{c} = \{w > x^{*}\}\). Using this, we
get \begin{align*}
\Pr(s X + d \in A^{c}) &= \Pr(sX + d \in \{w > x^{*}\}) \\
&= \Pr(sX + d > x^{*}) = \Pr(X > (x^{*} - d)/s).
\end{align*} Furthermore, \(P(X > w)\) is monotone and decreasing in
\(w\), which means if \(w \leq x^{*}\) and \(d \leq \Delta\),
\(\Pr(X > (w - \Delta)/s) \leq \delta\) implies
\(\Pr(X > (x^{*} - d)/s) \leq \delta\). The proof is concluded by noting
that the case for \(d = y-x\) also follows from the above. \qedhere

\end{proof}

\begin{remark}

Lemma \ref{lem:suffcritlc} above is restricted to even \(f\). An
extension without this restriction can be obtained by determining
conditions for cases \(d \geq 0\) and \(d \leq 0\) in the proof for
Lemma \ref{lem:suffcritlc} separately. Lemma \ref{lem:suffcritlc} would
then be a corollary of this extension.

\end{remark}

\begin{corollary}\label{cor:gausssuffcrit}

Let \(Z\) be a random variable distributed according to the standard
Gaussian distribution. Then for a real-valued function \(q\) on
databases with global sensitivity \(\Delta\) and a database \(d\), the
mechanism returning a variate of \(q(d) + \sigma Z\) is
\((\epsilon, \delta)\)-differentially private if
\begin{align}\label{eq:suffcrit}
\Pr \left( Z > \frac{\sigma  \epsilon}{\Delta}-\frac{\Delta}{2
\sigma} \right) \leq \delta.
\end{align}

\end{corollary}

\begin{proof}

The standard Gaussian density is \(\phi(x) = e^{-\psi(x)}\) for
\(\psi(x) = \frac{\log\left(2 \pi \right)}{2}+\frac{x^{2}}{2}\). It is
even and \(\psi\) is convex. The equation \[
\psi(x/\sigma) - \psi((x - \Delta)/\sigma) = \epsilon
\] has unique solution \[
x^{*} = \frac{\epsilon  \sigma^{2}}{{ \Delta }}+\frac{{ \Delta }}{2}
\] for \(\Delta\geq 0\). Using this solution yields \[
\frac{x^{*} - \Delta}{\sigma} = \frac{\sigma  \epsilon}{\Delta}-\frac{\Delta}{2
\sigma}.
\] The corollary then follows from Lemma \ref{lem:suffcritlc}. \qedhere

\end{proof}

\begin{remark}\label{rem:laplace}

For the standard Laplace distribution, \(\psi(x) = \log(2) + |x|\),
which is convex. Then \begin{align*}
\psi\left(\frac{x}{s}\right) - \psi\left(\frac{x-\Delta}{s}\right)  = 
{\left| \frac{x}{s}\right|}-{\left| \frac{\Delta-x}{s}\right|} = \epsilon
\end{align*} reduces to \(\frac{\Delta}{s} = \epsilon\) for
\(x \geq \Delta > 0\), and
\(\sup\{x \mid \psi\left(\frac{x}{\Delta/\epsilon}\right) - \psi\left(\frac{x-\Delta}{\Delta/\epsilon}\right) \leq \epsilon\} = \sup\{x \mid x \geq \Delta\} = \infty\).
Applying Lemma \ref{lem:suffcritlc} we conclude that for the standard
Laplace random variable \(X\), the mechanism that outputs a variate of
\(q(d) + \Delta/\epsilon X\) is \((\epsilon, 0)\)-differentially
private.

\end{remark}

\begin{lemma}\label{lem:suffcrit2}

Let \(Z\) be a random variable distributed according to the standard
Gaussian distribution. Then for \(\epsilon > 0\), \(\Delta > 0\), and
\(\delta \in (0,1)\) \begin{align}
\Pr \left( Z > \frac{\sigma  \epsilon}{\Delta}-\frac{\Delta}{2
\sigma} \right) \leq \delta, \tag{\ref{eq:suffcrit}}
\end{align} holds if and only if \(\sigma \geq b\) for \begin{align*}
b =
{\frac {{\Delta}}{2\,{\epsilon}} \left( \Phi^{-1} \left( 1-\delta \right) +\sqrt { \left( \Phi^{-1} \left( 1-\delta\right)  \right) ^{2}+2\,{\epsilon}} \right) }
\end{align*} where \(\Phi^{-1}\) is the standard Gaussian quantile
function.

\end{lemma}

\begin{proof}

Let \[
v(\sigma) = \frac{\sigma  \epsilon}{\Delta}-\frac{\Delta}{2 \sigma}.
\] Then, condition (\ref{eq:suffcrit}) can be written
\(\Pr(Z > v(\sigma)) \leq \delta\). Let \begin{align}
l(\sigma) &= \Pr\left(Z > v(\sigma)\right) = 1-\Phi(v(\sigma)) \notag\\
& \iff \notag\\
\Phi(v(\sigma)) &= 1 - l(\sigma) \notag\\
& \iff \notag\\ \label{eq:ld}
v(\sigma) &= \Phi^{-1}\left(1 - l(\sigma)\right).
\end{align} Recall that we want to find a lower bound for \(\sigma\)
such that \(l(\sigma) \leq \delta<1\). We note that \(l(\sigma)\) is
decreasing in \(\sigma\) if \(v\) is increasing in \(\sigma\). This is
the case since \[
v'(\sigma)={\frac {2\,{{\sigma}}^{2}{\epsilon}+{{\Delta}}^{2}}{2\,{\Delta}\,{{\sigma}}^{2}}}
\] is positive for all \(\sigma\) and \(\Delta>0\), \(\epsilon>0\).
Hence, we can find the sought lower bound by solving \(l(b) = \delta\)
for \(b\). We do this by substituting \(\delta\) for \(l(b)\) in
(\ref{eq:ld}) and solving for \(b > 0\), yielding \begin{align*}
b &= 
{\frac {{\Delta}}{2\,{\epsilon}} \left( \Phi^{-1} \left( 1-\delta
\right) +\sqrt { \left( \Phi^{-1} \left( 1-\delta \right)  \right)
^{2}+2\,{\epsilon}} \right) }. \qedhere
\end{align*}

\end{proof}

\begin{remark}\label{rem:folklore}

Let \(Z\) be a standard Gaussian random variable. Recall that
\(P(|Z| > x) \geq P(Z > x)\). This means that \(P(|Z| > x) \leq \delta\)
implies \(P(Z > x) \leq \delta\). Applying Corollary
\ref{cor:gausssuffcrit} we conclude that a sufficient condition for
adding Gaussian noise to achieve \((\epsilon, \delta)\)-differential
privacy is \begin{align} \label{eq:folklore}
\Pr\left(|Z| > \frac{\sigma  \epsilon}{\Delta}-\frac{\Delta}{2
\sigma}\right) \leq \delta.
\end{align} Inspecting proof of the standard bound given by Dwork and
Roth \autocite{dworkAlgorithmicFoundationsDifferential2014}, we note
that it is based on fulfilling the condition (\ref{eq:folklore}) above.
Replacing bound (\ref{eq:suffcrit}) by the bound (\ref{eq:folklore}) in
Lemma \ref{lem:suffcrit2} yields that we must have \begin{align*}
\sigma \geq b &= {\frac {{\Delta}}{2\,{\epsilon}} \left( \Phi^{-1} \left( 1-\delta/2
\right) +\sqrt { \left( \Phi^{-1} \left( 1-\delta/2 \right)  \right)
^{2}+2\,{\epsilon}} \right) } > \frac{\Delta}{\sqrt{2\epsilon}}. 
\end{align*} Since the above holds for all bounds fulfilling
(\ref{eq:folklore}), this represents a generalization and sharpening of
Theorem 4 in \autocite{pmlr-v80-balle18a} that states
\(\sigma \geq \frac{\Delta}{\sqrt{2\epsilon}}\) for the standard bound.

\end{remark}

\begin{lemma}\label{lem:qub}

Let \(\Phi^{{-1}}\) be the standard Gaussian quantile function. Then for
\(p \in (0,1)\) \begin{align*}
  \Phi^{-1}(p) \leq 
  \sqrt{2}\sqrt{\log\left(\frac{1}{4 p\left(1-p \right)}\right)}
  \cdot
  \begin{cases}
  -\frac{\sqrt{\pi}}{2}, & p < \frac{1}{2}, \\
  1, & p \geq \frac{1}{2}.
  \end{cases} 
\end{align*}

\end{lemma}

\begin{proof}

It is well known that
\(\operatorname{erf}(x) = \sign(x)P(\frac{1}{2}, x^{2})\), where \(P\)
is the regularized gamma function
\(P(s, x) = \frac{\gamma(s, x)}{\Gamma(s)}\) in which \(\Gamma\) and
\(\gamma\) are the Gamma and lower incomplete Gamma functions,
respectively (see, e.g., \autocite{olverNISTDigitalLibrary2020} 7.11.1).
From \autocite{olverNISTDigitalLibrary2020} (8.10.11) we have that
\[(1-e^{-\alpha_{a}x})^{a}\leq
P\left(a,x\right)\leq(1-e^{-\beta_{a}x})^{a}\] for \begin{align*}
\alpha_{a}&=\begin{cases}1,&0<a<1,\\
d_{a},&a>1,\end{cases}\\
\beta_{a}&=\begin{cases}d_{a},&0<a<1,\\
1,&a>1,\end{cases}\\
d_{a}&=(\Gamma\left(1+a\right))^{-1/a}.
\end{align*} Since \(a=1/2\) in our case, get that \begin{align*}
\operatorname{erf}(x) \geq 
\begin{cases}
-\sqrt{1-{e}^{-\frac{4 x^{2}}{\pi}}}, & x < 0, \\
\sqrt{1-e^{-x^{2}}}, & x\geq 0.
\end{cases} 
\end{align*} and consequently \begin{align*}
\operatorname{erf}^{-1}(x) \leq 
\begin{cases}
-\frac{\sqrt{\pi}}{2}  \sqrt{-\log\left(1-x^{2}\right)}, & x < 0, \\
\sqrt{-\log(1-x^2)}, & x \geq 0.
\end{cases} 
\end{align*} As
\(\Phi^{-1}(p) = \sqrt{2}\operatorname{erf}^{-1}(2p - 1)\), the Lemma
follows by substituting the upper bound for \(\operatorname{erf}^{-1}\).
\qedhere

\end{proof}

\begin{theorem}[Gaussian mechanism $(\epsilon, \delta)$-differential privacy]\label{thm:gaussmech}

Let \(q\) be a real valued function on databases with global sensitivity
\(\Delta\), and let \(Z\) be a standard Gaussian random variable. Then
for \(\delta \in (0,1)\) and \(\epsilon > 0\), the mechanism that
returns a variate of \(q(d) + \sigma Z\) is
\((\epsilon, \delta)\)-differentially private if \(\sigma \geq b\) where
\begin{align}
\label{eq:bboundex}
b &= 
{\frac {{\Delta}}{2\,{\epsilon}} \left( \Phi^{-1} \left( 1-\delta \right) +\sqrt { \left( \Phi^{-1} \left( 1-\delta\right)  \right) ^{2}+2\,{\epsilon}} \right) }\\
&\leq 
\frac{\Delta}{\epsilon\sqrt{2}} \left(\sqrt{a \log\left(\frac{1}{4
\delta  \left(1-\delta \right)}\right)+\epsilon}+s \sqrt{a
\log\left(\frac{1}{4 \delta  \left(1-\delta \right)}\right)}\right)
\label{eq:bbound}\\
&\leq 
\frac{\Delta}{\epsilon\sqrt{2}} (1 + s) \sqrt{a
\log\left(\frac{1}{4 \delta  \left(1-\delta \right)}\right)} + \frac{\Delta}{\sqrt{2\epsilon}},
 \label{eq:sbound} 
\end{align} where \(\Phi^{-1}\) is the standard Gaussian quantile
function and \begin{align*}
(a,s) &= 
\begin{cases}
(1,1), & 0 < \delta \leq \frac{1}{2}, \\
(\frac{\pi}{4}, -1), &  \frac{1}{2} \leq \delta < 1.
\end{cases} 
\end{align*}

\end{theorem}

\begin{proof}

Differential privacy and (\ref{eq:bboundex}) follows from Corollary
\ref{cor:gausssuffcrit} and Lemma \ref{lem:suffcrit2}. Applying the
bound for \(\Phi^{-1}\) from Lemma \ref{lem:qub} to
\(\Phi^{-1}(1-\delta)\) in (\ref{eq:bboundex}), yields the bound
(\ref{eq:bbound}) after some elementary manipulations. Bound
(\ref{eq:sbound}) is achieved by applying the fact that
\(\sqrt{a + b} \leq \sqrt{a} + \sqrt{b}\) for non-negative \(a,b\) to
the right hand side of (\ref{eq:bbound}). \qedhere

\end{proof}

\begin{remark}\label{rem:gaussball}

For a multidimensional statistic with \(\Delta\) determined using the
Euclidean norm, adding Gaussian noise with covariance matrix
\(\sigma^{2}\operatorname{diag}(1, 1, \ldots, 1)\) is
\((\epsilon,\delta)\)-differentially private for \(\sigma \geq b\) for
\(b\) given by (\ref{eq:bboundex}). This follows from Theorem
\ref{thm:gaussmech} and an argument Dwork and Roth use in their proof of
the standard bound (Theorem A.1.~in their monograph
\autocite{dworkAlgorithmicFoundationsDifferential2014}). This result was
also shown by Le Ny and Pappas \autocite{6606817}.

\end{remark}

\hypertarget{illustrating-constraints-of-the-standard-bound}{%
\section{Illustrating constraints of the standard
bound}\label{illustrating-constraints-of-the-standard-bound}}

Here we graphically illustrate that constraining \(\epsilon\) from above
for the standard bound is indeed needed. From Remark \ref{rem:folklore},
the sufficient condition for privacy the standard bound meets is
\begin{align*} 
\Pr(|Z| > v(\sigma, 2)) \leq \delta \tag*{(\ref{eq:folklore})}
\end{align*} where \[
v(\sigma, y) = \frac{\sigma \epsilon}{\Delta} - \frac{\Delta}{y\sigma}.
\] We further have that for \(s\) defined in (\ref{eq:drsigma}),
\begin{align*}
w(\epsilon, \delta) = v(s(\epsilon, \delta, \Delta), 2) &= 
{\frac {\sqrt {2}\left( 4\,\log  \left( {\frac {5}{4
\,\delta}} \right) -\epsilon \right)}{4 \sqrt {\log  \left( {\frac {5}{4\,\delta}} \right)}} 
}
\end{align*} which does not depend on \(\Delta\). Let \[
g(\epsilon, \delta) = \delta -
2(1 - \Phi(w(\epsilon, \delta))) = \delta - \Pr(|Z| >
w(\epsilon, \delta)).
\] Now, the sign of \(g\) determines whether the condition
(\ref{eq:folklore}) above is met. A plot of \(g(\epsilon, \delta)\) can
be seen in Figure \ref{fig:comp1a}. Interestingly, there exist
\(0< \delta < 1\) and \(0< \epsilon < 1\) such that (\ref{eq:folklore})
is violated as \(g(0.97, 0.97) < -0.005\), suggesting that technically a
constraint on \(\delta\) is needed to avoid violating
(\ref{eq:folklore}).

\begin{figure}
\centering
\begin{subfigure}[t]{.4\linewidth}
\centering\includegraphics[width=\textwidth]{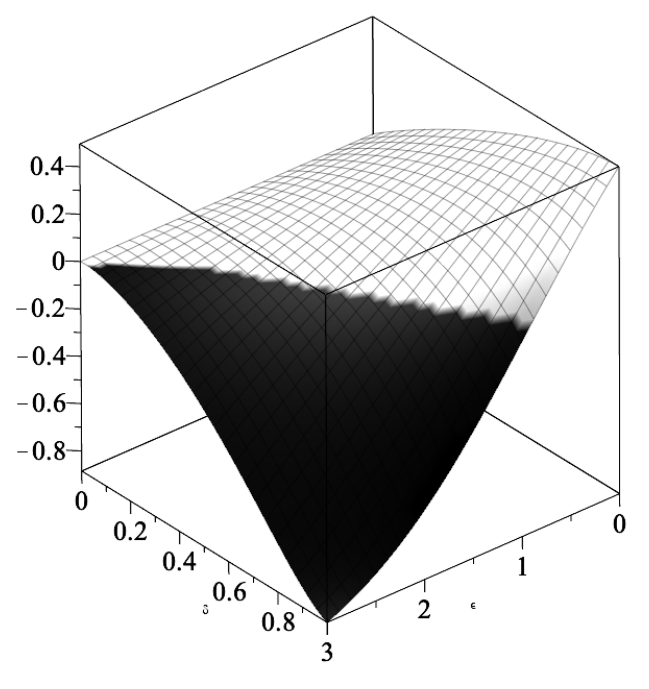}
\subcaption{}\label{fig:comp1a}
\end{subfigure} 
\hfill
\begin{subfigure}[t]{.4\linewidth}
\centering\includegraphics[width=\textwidth]{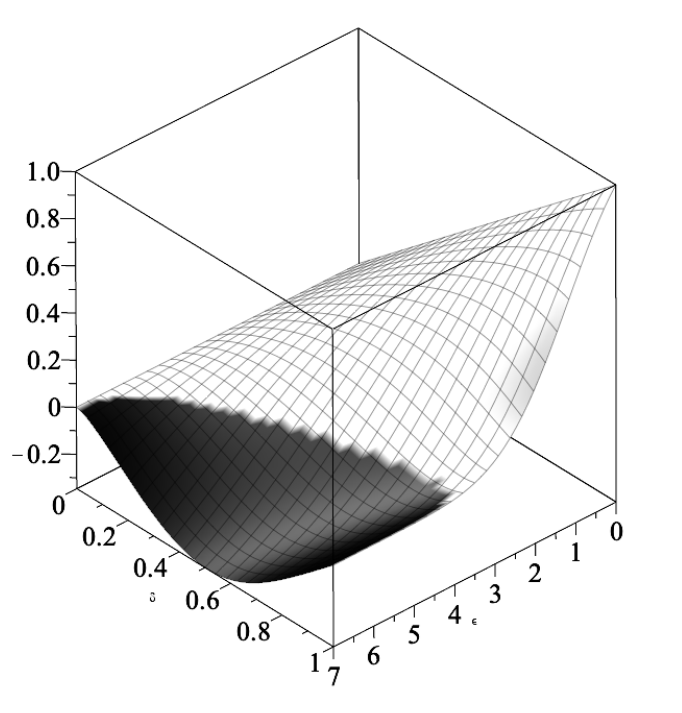}
\subcaption{}\label{fig:comp1b}
\end{subfigure}
\caption{\label{fig:comp1} (\subref{fig:comp1a}) $g(\epsilon, \delta)$. For the points $(\epsilon, \delta)$ where this
quantity is negative, indicated by dark gray, using
$\sigma=s(\epsilon, \delta, \Delta)$ violates the sufficient condition
(\ref{eq:folklore}). (\subref{fig:comp1b}) $d(\epsilon,\delta)$. Where
values are negative, indicated by dark gray, using
$\sigma=s(\epsilon, \delta, \Delta)$ violates condition
(\ref{eq:balle}) for $(\epsilon, \delta)$-differential privacy.}
\end{figure}

However, as Balle et al. \autocite{pmlr-v80-balle18a} point out,
violating (\ref{eq:folklore}) is is not the same as violating
\((\epsilon, \delta)\)-differential privacy. They show that
\((\epsilon,\delta)\)-differential privacy is achieved if and only if
\begin{align}
\label{eq:balle}
\Phi\left( \frac{\Delta}{2\sigma}-\frac{\epsilon\sigma}{\Delta} \right) - e^{\epsilon} \Phi\left(-\frac{\Delta}{2\sigma}-\frac{\epsilon\sigma}{\Delta}  \right) \leq \delta.
\end{align} They do not provide a closed form bound based on
(\ref{eq:balle}) but provide a numerical algorithm to compute the
smallest \(\sigma>0\) for which the above holds.

Substituting \(s(\epsilon, \delta, \Delta)\) for \(\sigma\) in the left
side of (\ref{eq:balle}), and subtracting this from \(\delta\) yields
\[d(\epsilon,\delta)=\delta - \left(\Phi(-v(s(\epsilon,\delta,
\Delta), 2)) - e^{\epsilon}\Phi(-v(s(\epsilon,\delta, \Delta),
-2))\right),
\] which does not depend on \(\Delta\). Analogous to \(g\) above, the
sign of \(d\) determines whether (\ref{eq:balle}) and
\((\epsilon, \delta)\)-differential privacy is violated. A plot of
\(d(\epsilon,\delta)\) can be seen in Figure \ref{fig:comp1b}. Negative
values indicate failure to be \((\epsilon, \delta)\)-differential
privacy. The plot suggests that even if the inequality of the standard
bound (\ref{eq:drsigma}) is strict, it is safe to consider it non-strict
for \(\epsilon\in (0,1)\). What the plot also shows, is that the
standard bound does not yield \((\epsilon, \delta)\)-differential
privacy for all \(\epsilon > 0\).

\hypertarget{comparing-the-two-bounds}{%
\section{Comparing the two bounds}\label{comparing-the-two-bounds}}

The standard bound and our bound differ in both being based on different
conditions and how closed form bounds are produced. Specifically,
(\ref{eq:folklore}) and via the Cramér--Chernoff style tail bound
\(\Pr(|Z| > x) \leq 2\phi(x)/x\), and (\ref{eq:suffcrit}) and via closed
form bound of the inverse error function leading to Lemma \ref{lem:qub},
respectively.

We now compare the two bounds for the common interval
\(\epsilon \in (0,1)\).

The ratio of the standard bound (\ref{eq:drsigma}) and our bound
(\ref{eq:bbound}) is \begin{align}
\label{eq:srat}
r(\epsilon, \delta) = 
\frac{2 \sqrt{\log\left(\frac{5}{4 \delta}\right)}}{\sqrt{a \log\left(\frac{1}{4 \delta  \left(1-\delta \right)}\right)+\epsilon}+s \sqrt{a \log\left(\frac{1}{4 \delta  \left(1-\delta \right)}\right)}}
\end{align} A value for \(r > 1\) means that the standard bound is
larger than ours. A plot of the ratio \(r\) can be seen in Figure
\ref{fig:comp2a}.

\begin{figure}[htb]
\centering
\begin{subfigure}[t]{.45\linewidth}
\centering\includegraphics[width=\textwidth]{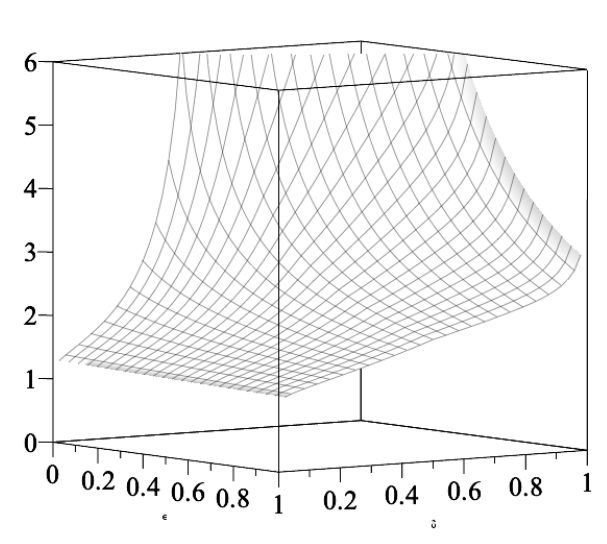}
\subcaption{}\label{fig:comp2a}
\end{subfigure} 
\hfill
\begin{subfigure}[t]{.45\linewidth}
\centering\includegraphics[width=\textwidth]{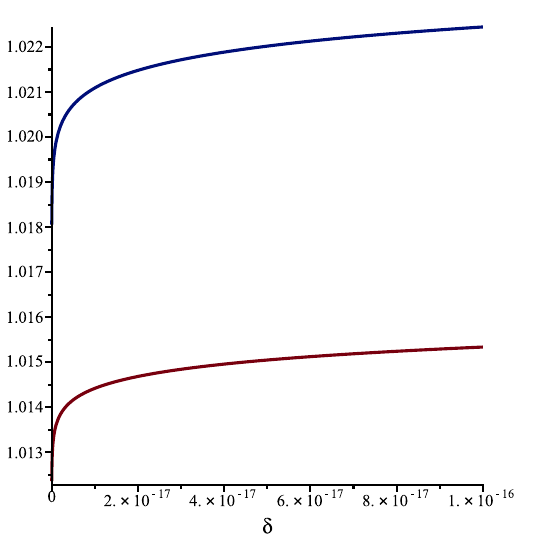}
\subcaption{}\label{fig:comp2c}
\end{subfigure}
\caption{\label{fig:comp2} (\subref{fig:comp2a}) The ratio $r$ given
in (\ref{eq:srat}). Note that all values are at
least 1. (\subref{fig:comp2c}) The functions $\rho_{+}(\delta) = r(0,
\delta)$ and $\rho_{-}(\delta) = r(1,\delta)$ giving upper and lower
bounds on the ratio $r$ in terms of $\delta$, respectively.}
\end{figure}

The partial derivative of \(r\) in (\ref{eq:srat}) with respect to
\(\epsilon\) is \begin{align}
\label{eq:sratdiff}
-\frac{\sqrt{\log\left(\frac{5}{4 \delta}\right)}}{\left(\sqrt{a \log\left(\frac{1}{4 \delta  \left(1-\delta \right)}\right)+\epsilon}+s \sqrt{a \log\left(\frac{1}{4 \delta  \left(1-\delta \right)}\right)}\right)^{2} \sqrt{a \log\left(\frac{1}{4 \delta  \left(1-\delta \right)}\right)+\epsilon}}.
\end{align} This derivative is negative for \(\delta > 0\) and
\(\epsilon > 0\), meaning that the ratio \(r\) decreases as \(\epsilon\)
increases. It can also be shown that partial derivative of \(r\) with
respect to \(\delta\) is \(O(\frac{1}{\delta})\).

We now look at what happens for small \(\delta \leq 1/2\). Inspecting
\(r\), we see that as \(\epsilon \to 0\) we get that
\(r \to \rho_{+}(\delta)\) where \[
\rho_{+}(\delta) = 
\frac{\sqrt{\log\left(\frac{5}{4 \delta}\right)}}{\sqrt{\log\left(\frac{1}{4 \delta  \left(1-\delta \right)}\right)}},
\] and as \(\epsilon \to 1\) we get that \(r \to \rho_{-}(\delta)\)
where \[
\rho_{-}(\delta) = 
\frac{2 \sqrt{\log\left(\frac{5}{4 \delta}\right)}}{\sqrt{\log\left(\frac{1}{4 \delta  \left(1-\delta \right)}\right)+1}+\sqrt{\log\left(\frac{1}{4 \delta  \left(1-\delta \right)}\right)}}.
\] Since \(r\) is decreasing in \(\epsilon\), the functions
\(\rho_{+}(\delta)\) and \(\rho_{-}(\delta)\) provide upper and lower
bounds on \(r\) for a given value of \(\delta \leq 1/2\). Both these
functions are increasing in \(0 < \delta \leq 1/2\) and as
\(\delta \to 0\) both go towards 1. A plot of \(\rho_{+}\) and
\(\rho_{-}\) for \(\delta \leq 10^{-16}\) can be seen in Figure
\ref{fig:comp2c}. As \(\rho_{+}(10^{-16}) < 1.023\), we see that for
small \(\delta\), the ratio \(r\) is not that big. In other words, while
our bound (\ref{eq:bbound}) is better than the standard bound, it is
only slightly better for \(\delta\) that can be considered small.

\hypertarget{discussion}{%
\section{Discussion}\label{discussion}}

Simple closed form bounds can be implemented using simple algorithms
with low implementation and computational complexity. The benefit of
this is a lower potential for errors, as well as decreasing power
consumption in low power devices whenever the alternative is using
iterative numerical algorithms to compute analytical solutions.
Furthermore, closed form relationships are useful in the analysis of
processes where privacy mechanisms are components or are applied
multiple times.

While our bound is better than the standard bound wherever this is
defined, our analysis suggests that for \(\delta\) that are small enough
to be considered relevant, the improvement is limited. Therefore, we
suggest that the main advantage of our bound is that it is valid for all
\(\epsilon>0\) and that it can be used as a drop in for the standard
bound without much difficulty even though it is slightly more complex.

We believe that the condition for \((\epsilon,\delta)\)-differential
privacy in Lemma \ref{lem:suffcritlc} is of independent interest as we
are able to derive the standard condition for \(\epsilon\)-differential
privacy in the case of Laplace noise.

Our bound (\ref{eq:bbound}) is based on the sufficient condition
(\ref{eq:suffcrit}). As Balle et al.~\autocite{pmlr-v80-balle18a}
demonstrate, the optimal \(\sigma\) can be gotten through numerically
optimizing the sufficient and necessary condition (\ref{eq:balle}). A
question we leave unaddressed for now is whether suitable closed form
bounds on \(\Phi\) can be substituted into (\ref{eq:balle}) to find an
even better closed form bound on \(\sigma\).

\printbibliography[title=References]

\renewcommand{\printbibliography}[1][]{}

\printbibliography[title=References]

\typeout{get arXiv to do 4 passes: Label(s) may have changed. Rerun}
\end{document}